\documentclass[12pt]{article}
\usepackage{sectsty}
\usepackage[T1]{fontenc}
\usepackage{kpfonts}
\usepackage[dvipsnames]{xcolor}
\usepackage{pdfpages}
\usepackage{sgame}
\usepackage{accents}
\usepackage{tikz-cd}
\usepackage{float}
\usepackage[round]{natbib}
\usepackage{multirow}
\usepackage{dutchcal}
\usepackage{pdfpages}
\usepackage{bbm}
\usepackage{sgame}
\usepackage{mathtools}
\usepackage{amsthm}
\usepackage{enumitem}
\usepackage[margin=1.25in]{geometry}
\usepackage{caption}
\usepackage{subcaption}
\usepackage{epigraph}
\usepackage{listings}
\usetikzlibrary{calc}
\usetikzlibrary{shapes,arrows}
\usepackage{tcolorbox}

\newtheorem{theorem}{Theorem}[section]
\newtheorem{proposition}[theorem]{Proposition}
\newtheorem{lemma}[theorem]{Lemma}
\newtheorem{corollary}[theorem]{Corollary}

\theoremstyle{definition}

\newtheorem{definition}[theorem]{Definition}

\newtheorem{remark}[theorem]{Remark}

\usepackage{hyperref}
\hypersetup{
    colorlinks=true,
    linkcolor=red,
    filecolor=red,      
    urlcolor=red,
    citecolor=blue,
}

\usepackage{setspace}

\definecolor{backcolour}{rgb}{0.63, 0.79, 0.95}

\lstdefinestyle{mystyle}{
  backgroundcolor=\color{backcolour},
  basicstyle=\ttfamily\footnotesize,
  breakatwhitespace=false,         
  breaklines=true,                 
  captionpos=b,                    
  keepspaces=true,                 
  numbers=left,                    
  numbersep=5pt,                  
  showspaces=false,                
  showstringspaces=false,
  showtabs=false,                  
  tabsize=2
}


\begin{document}
\author{Mark Whitmeyer \thanks{Arizona State University, \href{mailto:mark.whitmeyer@gmail.com}{mark.whitmeyer@gmail.com}. I thank Anujit Chakraborty, Rosemary Hopcroft, Ariel Rubinstein, Joseph Whitmeyer, and seminar audiences at Georgetown, Georgia, Indiana, NYU, and UCL for their comments. Draft date: \today.}}
\title{Calibrating the Subjective}
\maketitle

\begin{abstract}
I conduct \cite{rabin2000risk}'s calibration exercise in the subjective expected utility realm. I show that the rejection of some risky bet by a risk-averse agent only implies the rejection of more extreme and less desirable bets and nothing more.
\end{abstract}

%\keywords{Decision-making Under Uncertainty; Risk Aversion; Robust Predictions; Comparative Statics}\\
%\jel{D0; D8}

\setlength{\epigraphwidth}{4.9in}\begin{epigraphs}
\qitem{Einer hat immer Unrecht: aber mit zweien beginnt die Wahrheit. Einer kann sich nicht beweisen: aber zweie kann man bereits nicht widerlegen.}
{Friederich Nietzsche}
\qitem{Of course it is happening inside your head, Harry, but why on earth should that mean that it is not real?}%
      {Albus Dumbledore}
\end{epigraphs}

\section{Introduction}
``PROBABILITY DOES NOT EXIST,'' shouts De Finetti (\cite{de2017theory}) and, in doing so, rejects objective probability in favor of subjective. This approach, and more specifically the subjective expected utility (SEU) paradigm, is a behavioral definition of probability: it ``is a rate at which an individual is willing to bet on the occurrence of an event'' (\cite{nau2001finetti}). This stands in stark contrast to the objectivist view, in which probabilities are fundamental properties of events.

\cite{rabin2000risk}'s seminal calibration theorem demonstrates a striking pathology within the objective expected utility framework: the degree of risk aversion required to reject small-stakes gambles implies absurdly high aversion to larger-stakes gambles. Importantly, Rabin's critique presumes objective probabilities. What if ``PROBABILITIES DO NOT EXIST?"

This paper revisits the calibration puzzle through the lens of SEU. I conduct an analogous exercise--if a decision-maker (DM) prefers a sure thing to a risky gamble over some region of wealths, what are the other risky gambles that must be subjectively inferior to the sure thing?--and show that the pathologies identified by Rabin vanish in the subjective realm. I show that the only risky gambles that must be inferior to the sure thing are precisely those that are \textit{unambiguously worse} than the risky gamble that had originally been deemed inferior.

\cite{rabin2000risk} spurred an outpouring of papers scrutinizing the calibration issue. Notably, \cite{safra2008calibration} show that the issue is not idiosyncratic to expected utility: it persists in many settings in which the DM is not an expected utility maximizer. \cite{rubinstein2006dilemmas} argues that the paradox illustrates an issue with consequentialism. Others, e.g., \cite{cox2013} and \cite{bleichrodt2019resolving}, maintain that reference dependence is key to resolving the calibration issue.

Crucially, in these works that have followed, the probabilities remain, nevertheless, \textit{objective}. Here, I show that if this objectivity is not assumed, the extra freedom afforded by subjective probabilities is enough to avoid the extreme calibration predictions, despite the fact that my DM is an expected-utility maximizer with a utility function defined on her terminal wealth.

But won't the right collection of decisions pin down the DM's belief, in which case we are back in the objective setting? Not if the DM's utility is state-dependent, which I assume.\footnote{Note that with a suitable \textit{proxy}, one can cleverly elicit beliefs \citep{tsakas2025belief}.} What if the lotteries \textit{are} objective, as in, e.g., a laboratory setting? Here I appeal to De Finetti: why should the lotteries implied by the apparatus governing a DM's decision-making be the objective lotteries given to the DM?

%So, maybe \textit{that} is the problem.

%This result not only reconciles Rabin's findings with SEU but also highlights the flexibility and robustness of subjective beliefs in mitigating unrealistic predictions about risk preferences. The analysis offers fresh insights into the interplay between beliefs and preferences, suggesting a more nuanced interpretation of risk aversion in decision-making under uncertainty.

%The remainder of the paper is structured as follows: Section 2 formalizes the setup and key results, demonstrating the conditions under which SEU resolves the calibration paradox. Section 3 provides an in-depth proof of the main theorem, while Section 4 discusses implications for behavioral models of risk aversion and suggests directions for future research.

\section{The Formal Setting}

There are two binary-action menus, \(A = \left\{s,r\right\}\)--mnemomic for ``safe'' and ``risky''--and \(\hat{A} = \left\{s,\hat{r}\right\}\); and two states, \(\Theta = \left\{0,1\right\}\).\footnote{\label{footnote2}The binary setting is assumed merely for convenience--an analog of Theorem \ref{theorem} holds for general state spaces.} When faced with either menu, the decision-maker (DM) has a common subjective belief \(\mu \in \Delta\left(\Theta\right) = \left[0,1\right]\), where \(\mu \coloneqq \mathbb{P}(1)\) and a common \textit{state-dependent} risk-averse utility function in money \(u \colon \mathbb{R} \times \Theta \to \mathbb{R}\) that is strictly increasing and weakly concave in each state. \(\mathcal{U}\) denotes the class of such functions. \(u_\theta \colon \mathbb{R} \to \mathbb{R}\) denotes the DM's utility function in state \(\theta\). State dependence plays no role except to make my theorem stronger. It strengthens the sufficiency direction, and for my necessity proof by contraposition, I do not need it, constructing a state-independent utility function.

The DM's initial wealth is \(w \in \mathbb{R}\). The risk-free action \(s\) yields a state-independent monetary payoff of \(0\). The first risky action \(r\)'s monetary payoff is \(\alpha > 0\) in state \(1\) and \(-\beta < 0\) in state \(0\). Likewise, \(\hat{r}\) yields \(\hat{\alpha} > 0\) and \(-\hat{\beta} < 0\) in states \(1\) and \(0\). We assume that the DM is a subjective expected utility maximizer, preferring \(s\) to \(r\) if and only if \[\mu u_1(w) + (1-\mu)u_0(w) \geq \mu u_1(w+\alpha) + (1-\mu) u_0(w-\beta)\text{,}\] and \(s\) to \(\hat{r}\) if and only if \[\mu u_1(w) + (1-\mu)u_0(w) \geq \mu u_1(w+\hat{\alpha}) + (1-\mu) u_0(w-\hat{\beta})\text{.}\]
Suppressing the dependence on \(\mu\) and \(w\), we let \(s \succeq r\) represent the first inequality and \(s \succeq \hat{r}\) the second. \(\succ\) indicates the strict counterpart.

Suppose there exists a nonempty set of wealths \(W \subseteq \mathbb{R}\) such that at each \(w \in W\) the DM prefers \(s\) to \(r\), \textit{given her utility function and subjective belief}. What are the properties of \(\hat{r}\) such that the DM must also prefer \(s\) to \(\hat{r}\)?

\begin{definition}
    We say that \textcolor{red}{The Safe Option Must Remain Optimal} if, for all \(u \in \mathcal{U}\), \(s \succeq r\) for all \(w \in W\) implies \(s \succeq \hat{r}\) for all \(w \in W\).
\end{definition}
\begin{definition}\label{worse}
    We say that \textcolor{red}{The Risky Option Becomes Worse} if \(\hat{\beta} \geq \beta\), and an \textbf{Actuarial Worsening} transpires: \[\frac{\alpha}{\beta} \geq \frac{\hat{\alpha}}{\hat{\beta}}\tag{\(1\)}\label{in1}\text{.}\].
\end{definition}
\begin{theorem}\label{theorem}
    The safe option must remain optimal if and only if the risky option becomes worse.
\end{theorem}
Let us discuss the result before proving it formally. Without loss of generality we impose that \(0 \in W\), as we could just conduct this scaling within the DM's utility function. First, we note that an actuarial worsening is with respect to the belief at which a risk-neutral DM is indifferent between \(s\) and \(r\). Given this, it is clear that an actuarial worsening is necessary for the safe action to remain optimal: the class of risk-averse DMs includes those who are risk-neutral and so if the risky option strictly improves in an actuarial sense, there are beliefs close to a risk-neutral DM's indifference belief between \(s\) and \(r\) for which \(\hat{r} \succ s \succ r\)\textit{ for any \(w \in W\)}. 

Second, we observe that if an actuarial worsening transpires but \(\beta > \hat{\beta}\), it must be the case that \(\alpha > \hat{\alpha}\). This means that the new risky action \(\hat{r}\) is \textit{safer} (in the parlance of \cite{safety}) than \(r\); namely, more robust to increases in the DM's risk aversion. We then finish the necessity proof by completing the exercise in contraposition: we construct a \textit{state-independent} utility function that is
\begin{enumerate}
    \item continuous, strictly increasing, and concave on \(\mathbb{R}\), 
    \item kinked at \(\hat{\alpha}\) and \(-\hat{\beta}\), 
    \item linear on \(\left(-\hat{\beta},\hat{\alpha}\right)\), and
    \item of the constant absolute risk aversion class on \(\left[\hat{\alpha}, \infty\right]\) and \(\left[-\infty, -\hat{\beta}\right]\).
\end{enumerate}

The region of linearity means that the DM's indifference belief between \(s\) and \(\hat{r}\) when her wealth is \(0\) is \(\hat{\beta}/(\hat{\alpha} + \hat{\beta}) \). Crucially, the utility function we construct is parametrized in a way that that lets us scale the DM's risk-aversion up in the non-linear portions. Doing this scaling allows us to push the indifference belief between \(s\) and \(r\) to the right for all wealth values, making it so that, initially, the DM can be quite confident that the state is \(1\) yet still prefer \(s\) to \(r\). On the other hand, this confidence means that when she is picking between \(s\) and \(\hat{r}\), the DM prefers \(\hat{r}\). In short, by scaling the risk aversion we can find a belief such that \(s \succ r\) for all \(w \in \mathbb{R}\) yet \(\hat{r} \succ s\) for \(w = 0\), yielding the result.

The sufficiency direction is an easy chain of inequalities and is similar to the first theorem in \cite{hmab}, although we prove that theorem in a different manner, as the utilities there are \textit{state-independent}.

%\subsection{Two Extensions}

Footnote \ref{footnote2} claims that it is easy to extend Theorem \ref{theorem} to general state spaces. Here is an informal discussion of how. Let \(\Theta\) be a compact subset of the real numbers and let \(\mathcal{R}\), \(\mathcal{S}\), and \(\mathcal{T}\) denote the sets of states in which the risky action's (\(r\)'s) payoff is strictly higher than \(0\), strictly lower than \(0\), and equal to \(0\). The safe action, \(s\), yields a state-independent monetary payoff of \(0\); whereas the risky actions \(r\) and \(\hat{r}\) yield payoffs of \(\alpha_\theta, \hat{\alpha}_\theta > 0\) in each \(\theta \in \mathcal{R}\), \(-\beta_\theta, -\hat{\beta}_\theta < 0\) in each \(\theta \in \mathcal{S}\), and \(0\) in each \(\theta \in \mathcal{T}\).

\begin{remark}
    The safe option must remain optimal if and only if \(\hat{\beta}_{\theta'} \geq \beta_{\theta'}\) for all \(\theta' \in \mathcal{S}\) and \(\alpha_{\theta}/\beta_{\theta'} \geq \hat{\alpha}_{\theta}/\hat{\beta}_{\theta'}\) for all \(\theta \in \mathcal{R}\), for all \(\theta' \in \mathcal{S}\).
\end{remark}
This remark is an easy consequence of the fact that the extreme points of the set of beliefs for which \(s \succeq r\) are the binary distributions supported on states in \(\mathcal{S}\) and \(\mathcal{R}\) and the degenerate distributions on states in \(\mathcal{S} \cup \mathcal{T}\). Consequently, all that is needed is to aggregate pairwise the conditions in Theorem \ref{theorem}.

\iffalse
It is also straightforward to see that an analog of Theorem \ref{theorem} holds for a generalization of SEU. For simplicity, let us return to the binary state environment. There, we understand an action as an element of \(\mathbb{R}^\Theta\) and assume that a DM's preferences over actions are represented by an increasing quasiconcave function \(V \colon \mathbb{R}^\Theta \to \mathbb{R}\), where the DM prefers action \(a\) to \(b\) if and only if \(V(a) \geq V(b)\). We call such preferences \textcolor{Aquamarine}{Convex}. Then, as the necessity direction is implied by Theorem \ref{theorem}, we have the following corollary:
\begin{corollary}
    If the DM has convex preferences, \(s \succeq r\) for all \(w \in W\) implies \(s \geq \hat{r}\) for all \(w \in W\) if and only if the risky option becomes worse.
\end{corollary}
\begin{proof}
    Let the risky option become worse. Observe that this implies there exists \(\lambda \in \left(0,1\right]\) such that \(r \geq \lambda \hat{r} + (1-\lambda) s\) (as vectors in \(\mathbb{R}^\Theta\)). Suppose for the sake of contradiction that there exists \(w \in W\) such that \(\hat{r} \succ s \succeq r\). But then, taking \(w = 0\) without loss of generality,
    \[V(\hat{r}) > V(s) \geq V(r) \geq V(\lambda \hat{r} + (1-\lambda) s) \geq \]
\end{proof}\fi

\section{Proof of Theorem \ref{theorem}}
\noindent \(\left(\Rightarrow\right)\) If there is not an actuarial worsening (Inequality \ref{in1} does not hold), we are done, as there will be subjective beliefs such that for all \(w \in W\), \(\hat{r} \succ s \succ r\) for a risk-neutral DM. So, let Inequality \ref{in1} hold but suppose for the sake of contraposition that \(\beta > \hat{\beta}\), which implies \(\alpha > \hat{\alpha}\).

Now, we construct a state-independent utility function as follows. For \(k \geq 1\), and dropping the subscript from \(u\) due to state-independence, define
\[u(x) \coloneqq \begin{cases}
-\hat{\beta}+\exp\left(k\hat{\beta}\right)-\exp\left(-kx\right) \quad &\text{if} \quad x \leq -\hat{\beta}\\
x \quad &\text{if} \quad -\hat{\beta} < x < \hat{\alpha}\\
\hat{\alpha}+\exp\left(-k\hat{\alpha}\right)-\exp\left(-kx\right) \quad &\text{if} \quad \hat{\alpha} \leq x\text{.}
\end{cases}\]
By construction, \(u\) is continuous, strictly increasing, and weakly concave on \(\mathbb{R}\).\footnote{One can even adapt the construction so that \(u\) is differentiable.} Moreover, when \(w = 0\) the DM's indifference belief under menu \(\left\{s,\hat{r}\right\}\) is \[\hat{\mu}^* \coloneqq \frac{\hat{\beta}}{\hat{\beta} + \hat{\alpha}}\text{.}\]

When \(w \geq \hat{\alpha} + \beta\) or \(w \leq -\hat{\beta}-\alpha\), the DM's indifference belief under menu \(\left\{s,r\right\}\) is
\[\overline{\mu}_k \coloneqq \frac{e^{\alpha k} \left(e^{\beta k} - 1\right)}{e^{\left(\alpha + \beta\right) k} - 1}\text{.}\]
Importantly, \(\overline{\mu}_k\) is increasing in \(k\) and converges to \(1\) as \(k \to \infty\). There are seven other possible regions in which \(w\) can lie. Leaving the details to Appendix \ref{theoremproof}, we show that for any wealth in each region, the DM's indifference belief \(\mu^{i}_k\) (\(i \in \left\{1,\dots,7\right\}\)) is strictly larger than \(\hat{\mu}^*\) provided \(k\) is sufficiently large--in fact, in all but one region, like \(\overline{\mu}_k\), \(\mu^{i}_k \to 1\) as \(k \to \infty\). Consequently, if \(k\) is sufficiently large, there is a belief \(\mu\) such that for all \(w \in W\), \(s \succ r\), yet for \(w = 0\), \(\hat{r} \succ s\). 

\medskip

\noindent \(\left(\Leftarrow\right)\) Please see Appendix \ref{theoremproof}. \hfill \(\blacksquare\)

\section{When the Ranking Must Flip}

We finish with a result concerning situations in which the DM prefers \(s\) to \(r\) for all \(w \in W\) but strictly prefers \(\hat{r}\) to \(s\) \textit{for all} \(w \in W\).

\begin{proposition}
    If the risky option does not become worse, there exists a state-independent \(u \in \mathcal{U}\) and a nondegenerate interval \(\left[\underline{w},\overline{w}\right]\) such that \(\hat{r} \succ s \succeq r\) for all \(w \in \left[\underline{w},\overline{w}\right]\).
\end{proposition}
\begin{proof}
    As discussed above, if an actuarial worsening does not transpire, we can find a belief such that \(\hat{r} \succ s \succ r\) for a risk-neutral DM. So, suppose instead that an actuarial worsening happens but that \(\beta > \hat{\beta}\) and \(\alpha > \hat{\alpha}\). Take an arbitrary nondegenerate interval \(\left[\underline{w},\overline{w}\right]\) with \(\overline{w} -  \underline{w} < \beta - \hat{\beta}\); and define
    \[u(x) \coloneqq \begin{cases}
        x, \quad &\text{if} \quad x < \underline{w} - \hat{\beta}\\
        \iota x + (1-\iota) \left(\underline{w} - \hat{\beta}\right), \quad &\text{if} \quad x \geq \underline{w} - \hat{\beta}\text{,}
    \end{cases}\]
    for some \(\iota \in \left(0,1\right]\).

    Then, for all \(w \in \left[\underline{w},\overline{w}\right]\), the indifference belief between \(s\) and \(\hat{r}\) is \(\hat{\beta}/(\hat{\alpha}+\hat{\beta})\). On the other hand, for all \(w \in \left[\underline{w},\overline{w}\right]\), the indifference belief between \(s\) and \(r\) is
    \[\frac{\iota w + (1-\iota) \left(\underline{w} - \hat{\beta}\right) - (w - \beta)}{\iota w + (1-\iota) \left(\underline{w} - \hat{\beta}\right) - (w - \beta) + \iota \alpha}\text{,}\]
    which is strictly decreasing in \(\iota\) and equals \(1\) as \(\iota \downarrow 0\). 
    
    Consequently, there exists \(u \in \mathcal{U}\) and a belief \(\mu\) such that \(\hat{r} \succ s \succ r\).\end{proof}

\appendix

\section{Completion of Theorem \ref{theorem}'s Proof}\label{theoremproof}

We need to check that for all sufficiently large \(k\), for any \(w \in W\), the DM's indifference belief between \(s\) and \(r\) is strictly larger than \(\hat{\beta}/(\hat{\beta} + \hat{\alpha})\). We have already verified this for extreme wealths, but now need to do so for intermediate ones. The DM's indifference beliefs to be computed are for menu \(\left\{s,r\right\}\) and the formula is
\[\frac{u(w) - u(w-\beta)}{u(w+\alpha) - u(w-\beta)}\text{.}\]

\medskip

\noindent \textbf{Cases 1 \& 2.} When \(-\hat{\beta} \geq w > -\hat{\beta} - \alpha\), the indifference belief is
\[\mu_k^1 \coloneqq \frac{-\exp\left(-kw\right) + \exp\left(-k (w - \beta)\right)}{w + \alpha + \hat{\beta} - \exp\left(k\hat{\beta}\right) +\exp\left(-k(w-\beta)\right)}\text{,}\] if 
\(w + \alpha \leq \hat{\alpha}\); and it is
\[\mu_k^2 \coloneqq \frac{-\exp\left(-kw\right) + \exp\left(-k (w - \beta)\right)}{\hat{\alpha}+\exp\left(-k\hat{\alpha}\right)-\exp\left(-k(w+\alpha)\right) + \hat{\beta} - \exp\left(k\hat{\beta}\right) +\exp\left(-k(w-\beta)\right)}\text{,}\]
if
\(w + \alpha \geq \hat{\alpha}\). 

\(\mu^1_k\) simplifies to
\[-\frac{e^{\beta k}-1}{\left(e^{\hat{\beta} k}-w-\hat{\beta}-\alpha \right)e^{w k}-e^{\beta k}}\text{,}\]
which is larger than \(1 - \exp\left\{-\beta k\right\}\) for all sufficiently large \(k\). Thus, as \(k \to \infty\), \(\mu^1_k \to 1\). 

\(\mu^2_k\) simplifies to
\[\frac{1-e^{-\beta k}}{1-e^{\left(\hat{\beta}+w-\beta\right)k}+\left(\hat{\beta}+\hat{\alpha}\right)e^{-\left(\beta-w\right)k}+e^{\left(w-\hat{\alpha}-\beta\right)k}-e^{-\left(\alpha+\beta\right)k}}\text{.}\]
Both the numerator and the denominator converge to \(1\) as \(k \to \infty\), so \(\mu^2_k\) does as well.

\medskip

\noindent \textbf{Cases 3 \& 4.} When \(\hat{\alpha} \leq w < \hat{\alpha} + \beta\), the indifference belief is
\[\mu_k^3 \coloneqq \frac{\hat{\alpha}+\exp\left(-k\hat{\alpha}\right)-\exp\left(-kw\right) - (w-\beta)}{\hat{\alpha}+\exp\left(-k\hat{\alpha}\right)-\exp\left(-k(w+\alpha)\right) - (w-\beta)}\text{,}\] if 
\(w - \beta \geq -\hat{\beta}\); and it is
\[\mu_k^4 \coloneqq \frac{\hat{\alpha}+\exp\left(-k\hat{\alpha}\right)-\exp\left(-kw\right) + \hat{\beta} - \exp\left(k\hat{\beta}\right) +\exp\left(-k(w-\beta)\right)}{\hat{\alpha}+\exp\left(-k\hat{\alpha}\right)-\exp\left(-k(w+\alpha)\right) + \hat{\beta} - \exp\left(k\hat{\beta}\right) +\exp\left(-k(w-\beta)\right)}\text{,}\]
if
\(w - \beta < -\hat{\beta}\).

\(\mu^3_k\) simplifies to
\[\frac{\exp\left(-k \hat{a}\right)-\frac{1}{\exp\left(kw\right)}+ \hat{a}-(w-\beta)}{\exp\left(-k \hat{a}\right)-\frac{1}{\exp\left(k\left(w+\alpha\right)\right)}+ \hat{a}-(w-\beta)}\text{,}\]
which converges to \(1\) as \(k \to \infty\).

\(\mu^4_k\) simplifies to
\[\frac{\frac{\hat{\alpha} + \hat{\beta}}{\exp\left(-k(w-\beta)\right)}+\frac{1}{\exp\left(-k(w-\beta - \hat{\alpha})\right)}-\frac{1}{\exp\left(k\beta)\right)}-\frac{1}{\exp\left(-k(w-\beta+\hat{\beta})\right)}+1}{\frac{\hat{\alpha} + \hat{\beta}}{\exp\left(-k(w-\beta)\right)}+\frac{1}{\exp\left(-k(w-\beta - \hat{\alpha})\right)}-\frac{1}{\exp\left(k(\alpha+\beta)\right)}-\frac{1}{\exp\left(-k(w-\beta+\hat{\beta})\right)}+1}\text{,}\]
which converges to \(1\) as \(k \to \infty\).

\medskip

\noindent \textbf{Cases 5, 6, \& 7.} When \(-\hat{\beta} \leq w - \beta\) and \(w \leq \hat{\alpha} < w + \alpha\), the indifference belief is
\[\mu_k^5 \coloneqq \frac{\beta}{\hat{\alpha}+\exp\left(-k\hat{\alpha}\right)-\exp\left(-k(w+\alpha)\right) - (w - \beta)} \to \frac{\beta}{\hat{\alpha}-w+\beta}\text{,}\]
as \(k \to \infty\). Moreover, 
\[\frac{\beta}{\hat{\alpha}-w+\beta} > \frac{\hat{\beta}}{\hat{\alpha}+\hat{\beta}} \quad \Leftrightarrow \quad \hat{\alpha}\left(\beta - \hat{\beta}\right) + \hat{\beta} w > 0\text{,}\]
which is true.

If \(w-\beta < -\hat{\beta} \leq w\) and \(w+\alpha \leq \hat{\alpha}\),
\[\mu_k^6 \coloneqq \frac{w + \hat{\beta} - \exp\left(k\hat{\beta}\right) +\exp\left(-k(w-\beta)\right)}{w+\alpha + \hat{\beta} - \exp\left(k\hat{\beta}\right) +\exp\left(-k(w-\beta)\right)} \to 1\text{,}\]
as \(k \to \infty\).

Finally, if \(w-\beta < -\hat{\beta} \leq w \leq \hat{\alpha} < w + \alpha\),
\[\mu_k^7 \coloneqq \frac{w + \hat{\beta} - \exp\left(k\hat{\beta}\right) +\exp\left(-k(w-\beta)\right)}{\hat{\alpha}+\exp\left(-k\hat{\alpha}\right)-\exp\left(-k(w+\alpha)\right) + \hat{\beta} - \exp\left(k\hat{\beta}\right) +\exp\left(-k(w-\beta)\right)} \to 1\text{,}\]
as \(k \to \infty\).

Here is the sufficiency direction.
\begin{lemma}
    If the risky option becomes worse, the safe option must remain optimal.
\end{lemma}
\begin{proof}
    Let \(\hat{\beta} \geq \beta\) and \(\alpha/\beta \geq \hat{\alpha}/\hat{\beta}\). If \(\alpha \geq \hat{\alpha}\), \(r\) weakly dominates \(\hat{r}\), so for all \(w \in W\), we must have \(s \succeq r \succeq \hat{r}\). If \(\alpha < \hat{\alpha}\), for all \(w \in W\), starting with the indifference belief between \(s\) and \(\hat{r}\), we have
    \[\begin{split}
        \frac{u_0\left(w\right) - u_0\left(w-\hat{\beta}\right)}{u_0\left(w\right) - u_0\left(w-\hat{\beta}\right)+u_1\left(w+\hat{\alpha}\right) - u_1\left(w\right)} &= \frac{\frac{u_0\left(w\right) - u_0\left(w-\hat{\beta}\right)}{w - \left(w-\hat{\beta}\right)}}{\frac{u_0\left(w\right) - u_0\left(w-\hat{\beta}\right)}{w - \left(w-\hat{\beta}\right)} + \frac{u_1\left(w+\hat{\alpha}\right) - u_1\left(w\right)}{w - \left(w-\hat{\beta}\right)}}\\
        &\geq \frac{u_0\left(w\right) - u_0\left(w-\beta\right)}{u_0\left(w\right) - u_0\left(w-\beta\right) + \frac{\beta}{\hat{\beta}}\hat{\alpha}\frac{u_1\left(w+\hat{\alpha}\right) - u_1\left(w\right)}{w+\hat{\alpha}-w}}\\
        &\geq \frac{u_0\left(w\right) - u_0\left(w-\beta\right)}{u_0\left(w\right) - u_0\left(w-\beta\right) + \alpha\frac{u_1\left(w+\hat{\alpha}\right) - u_1\left(w\right)}{w+\hat{\alpha}-w}}\\
        &\geq \frac{u_0\left(w\right) - u_0\left(w-\beta\right)}{u_0\left(w\right) - u_0\left(w-\beta\right)+u_1\left(w+\alpha\right) - u_1\left(w\right)}\text{,}
    \end{split}\]
    which is the indifference belief between \(s\) and \(r\);
    where the first and third inequalities follow from the Three-chord lemma (Theorem 1.16 in \cite*{phelps2009convex}), and the second inequality from Inequality \ref{in1}.
\end{proof}

\bibliography{sample.bib}

\end{document}